\documentclass{IEEEtaes}
\usepackage{cite}
\usepackage{amsmath,amssymb,amsfonts}
\usepackage{algorithmic}
\usepackage{graphicx}
\usepackage{textcomp}
\usepackage{xcolor}
\usepackage{color,array,amsthm, amsmath, amsfonts, amssymb, algorithmic, algorithm, bbm, url,array, comment}
\usepackage{graphicx}
\usepackage{subcaption}
\usepackage{url}
\usepackage[acronyms]{glossaries}

\jvol{XX}
\jnum{XX}
\jmonth{XXXXX}
\paper{1234567}
\pubyear{2022}
\doiinfo{TAES.2022.Doi Number}

\newtheorem{theorem}{Theorem}

\newtheorem{assumption}{Assumption}
\newtheorem{definition}{Definition}

\newcommand{\Complex}{\mathbb{C}}
\newcommand{\E}{\mathbb{E}}
\newcommand{\Cov}{\boldsymbol{\Sigma}}
\newcommand{\I}{\mathbf{I}}

\newcommand{\SampCov}{\mathbf{S}}

\newcommand{\banded}{\mathbf{B}}
\newcommand{\lowerbound}{{\underline{F}}}
\newcommand{\upperbound}{{\overline{F}}}
\newcommand{\diag}{{\text{diag}}}
\newcommand{\matrixdimension}{{p}}
\newcommand{\const}{\mathrm{const.}}
\newcommand{\VEC}{\mathrm{vec}}
\newcommand{\argmin}{\mathrm{argmin}}
\newcommand{\conv}{\mathrm{conv}}

\def\half{\frac{1}{2}}

\def\Al{\mathbf{A}^{(\ell)}}

\def\gl{{g_\ell}}

\def\tr{\mathrm{tr}}
\def\E{\mathbb{E}}

\def\S{\mathbf S}
\def\hatP{\hat{\mathbf{B}}}

\newacronym{stap}{STAP}{Space-Time Adaptive Processing}
\newacronym{scm}{SCM}{Sample Covariance Matrix}
\newacronym{scnr}{SCNR}{Signal--to--Clutter--plus--Noise--Ratio}
\newacronym{mp}{MP}{Marchenko--Pastur}
\newacronym{mle}{MLE}{Maximum Likelihood Estimator}
\newacronym{rmt}{RMT}{Random Matrix Theory}
\newacronym{svd}{SVD}{Singular Value Decomposition}
\newacronym{bcd}{BCD}{Block Co-ordinate Descent}
\newacronym{tobasco}{TABASCO}{TApered or BAnded Shrinkage COvariance}
\newacronym{icm}{ICM}{Internal Clutter Motion}
\newacronym{kkt}{KKT}{Karush-Kuhn-Tucker}
\newacronym{crb}{CRB}{Cramer-Rao Bound}

\def\BibTeX{{\rm B\kern-.05em{\sc i\kern-.025em b}\kern-.08em
    T\kern-.1667em\lower.7ex\hbox{E}\kern-.125emX}}

\begin{document}

\title{Approximate MLE of High-Dimensional STAP Covariance Matrices with Banded \& Spiked Structure -- A Convex Relaxation Approach} 

\author{Shashwat Jain}
\member{Student Member, IEEE}
\affil{Cornell University, Ithaca, NY, USA} 

\author{Vikram Krishnamurthy}
\member{Fellow, IEEE}
\affil{Cornell University, Ithaca, NY, USA} 

\author{Muralidhar Rangaswamy}
\member{Fellow, IEEE}
\affil{Air Force Research Laboratory, Wright Patterson AFB, OH, USA}

\author{Sandeep Gogineni}
\member{Senior Member, IEEE}
\affil{Information Systems Laboratories Inc., Dayton, OH, USA}

\author{Bosung Kang}
\member{Member, IEEE}
\affil{University of Dayton Research Institute, Dayton, OH, USA}

\author{Sean M. O'Rourke}
\member{Member, IEEE}
\affil{U.S. Army Combat Capabilities Development Command, Army Research Laboratory, Adelphi, MD, USA}

%% \author{FOURTH D. AUTHOR}
%% \affil{University of Colorado, Colorado, USA}

\receiveddate{ }
%% \accepteddate{XXXXX XX XXXX}
%% \publisheddate{XXXXX XX XXXX}

%\corresp{Corresponding Author: Shashwat Jain}

\authoraddress{Shashwat Jain: sj474@cornell.edu, Vikram Krishnamurthy: vikramk@cornell.edu, Muralidhar Rangaswamy: muralidhar.rangaswamy@us.af.mil, Sandeep Gogineni: sgogineni@islinc.com, Bosung Kang: bosung.kang@udri.udayton.edu and Sean M. O'Rourke: sean.m.orourke18.civ@army.mil.}

\editor{}
\supplementary{For reproducibility of the results, the code is available at 
\hyperlink{https://github.com/sjain474/BandedPlusSpiked.git}{\texttt{GitHub}}. 
A shorter preliminary version of this paper titled ``MLE of  Banded, Spiked, High--Dimensional Clutter Covariance--A Convex Relaxation Approach" is to appear in the proceedings of the International Radar Conference 2025.\\ 
This research was funded by Army Research Office grant W911NF-24-1-0083 and NSF grant CCF-2312198. Distribution Statement A. Approved for public release: distribution is unlimited.}

\markboth{Jain et. al.}{Banded Plus Spiked MLE}
\maketitle

\begin{abstract}
Estimating the clutter-plus-noise covariance matrix in high-dimensional \gls*{stap} is challenging in the presence of Internal Clutter Motion (ICM) and a high noise floor. The problem becomes more difficult in low-sample regimes, where the Sample Covariance Matrix (SCM) becomes ill-conditioned.  To capture the ICM and high noise floor, we model the covariance matrix using a ``Banded+Spiked'' structure. Since the Maximum Likelihood Estimation (MLE) for this model is non-convex, we propose a convex relaxation which is formulated as a Frobenius norm minimization with non-smooth convex constraints enforcing banded sparsity. This relaxation serves as a provable upper bound for the non-convex likelihood maximization and extends to cases where the covariance matrix dimension exceeds the number of samples. We derive a variational inequality-based bound to assess its quality. We introduce a novel algorithm to jointly estimate the banded clutter covariance and noise power. Additionally, we establish conditions ensuring the estimated covariance matrix remains positive definite and the bandsize is accurately recovered. Numerical results using the high-fidelity RFView radar simulation environment demonstrate that our algorithm achieves a higher Signal--to--Clutter--plus--Noise Ratio (SCNR) than state-of-the-art methods, including TABASCO, Spiked Covariance Stein Shrinkage, and Diagonal Loading, particularly when the covariance matrix dimension exceeds the number of samples.
\end{abstract}

\begin{IEEEkeywords}
Clutter--Plus--Noise Covariance Estimation, Banded Matrix, Spiked Covariance Matrix, MLE, Convex Relaxation
\end{IEEEkeywords}

\section{INTRODUCTION}
In high-dimensional Space-Time Adaptive Processing (\gls*{stap}), estimating the clutter-plus-noise covariance matrix is particularly challenging when the number of samples is limited relative to the matrix dimension. Traditional methods, such as the \gls*{scm}, often become ill-conditioned and exhibit high estimation error in these low-sample regimes, leading to significantly reduced \gls*{scnr} and degraded beamforming and target detection performance. The presence of \gls*{icm} with a high noise floor further exacerbates the estimation. Our approach models the clutter--plus--noise covariance matrix as a banded structure, where most eigenvalues lie below the noise power, resulting in a ``Banded+Spiked'' model. To handle the non-convex nature of the log-likelihood function used to compute the \gls*{mle}, we introduce a convex relaxation and propose an estimation algorithm. The proposed algorithm yields higher \gls*{scnr} under limited sample conditions, outperforming the existing state--of--the--art methods. 
\subsection*{Why Banded+Spiked Covariance Structure?}
The ``Banded" structure in a clutter covariance matrix models local spatio-temporal dependencies by concentrating non-zero elements near the main diagonal within a bandsize of $L$ sub-diagonals, leaving distant off-diagonal elements significantly smaller in comparison.  This phenomenon is documented clearly in \cite{ward1998space} when \gls*{icm} is present.  This structure is particularly important for airborne radar, where successive radar pulses are spatio-temporally correlated over short intervals but not over long ones, creating a natural sparsity pattern. However, while the banded structure captures sparsity, it does not account for the noise-dominated distribution of singular values. This is where the ``Spiked" component is essential to model the clutter plus noise covariance matrix. In \gls*{stap} applications, where thermal noise dominates, most singular values of the \gls*{scm} are influenced by the noise floor, and follow the \gls*{mp} distribution in high dimensions. By leveraging this ``Banded+Spiked" model, we capture both the sparsity and the noise properties, for accurate clutter--plus--noise covariance estimation. This model is suitable when clutter--plus--noise covariance matrix has a high noise floor in the presence of \gls*{icm}.\\
\subsection*{Related Works}
Covariance matrix estimation in data-deficient scenarios has been a significant area of research in radar signal processing, particularly for high-dimensional settings where the \gls*{scm} becomes ill-conditioned and unreliable with limited samples \cite{ReedMallet}. Traditional methods, such as diagonal loading \cite{abramovich1981analysis, abramovich1981controlled, carlson1988covariance,Ward1994} and factored space-time methods \cite{dipietro1992extended}, have been used to address the ill-conditioning of the \gls*{scm}. Additional adaptive methods include Principal Components Inverse \cite{kirsteins1994adaptive}, the Multistage Wiener Filter \cite{goldstein1998multistage}, the Parametric Adaptive Matched Filter \cite{roman2000parametric}, and the EigenCanceler \cite{haimovich1996eigencanceler}, along with data-independent approaches such as JDL-GLR \cite{wang1994adaptive}. \gls*{rmt} provides theoretical insights into \gls*{scm} behavior, enabling shrinkage estimators to stabilize covariance estimates in high-dimensional, data-limited settings \cite{wainwright2019high, vershynin2018high, bai2010spectral, wang2017asymptotics}. %Popular shrinkage methods include the Ledoit-Wolf estimator \cite{ledoit2004well}, regularized PCA \cite{bun2016rotational}, ridge and lasso shrinkage \cite{yuan2020improved}, and regularized M-estimators \cite{couillet2014robust}. These methods have found applications in wireless communications \cite{tulino2004random}, direction of arrival estimation \cite{couillet2013joint}, and financial portfolio optimization \cite{ledoit2022quadratic}.

In the absence of \gls*{icm}, clutter covariance matrices frequently exhibit a low-rank structure, and several estimation methods assume the rank-deficient clutter covariance matrix~\cite{sen2015low, spClut1, kang2016expected, abrahamsson2007enhanced}. Under ideal conditions, Brennan’s rule~\cite{Brennan1973} gives the rank of the clutter covariance matrix, but it often fails under real-world complexities, such as \gls*{icm} and mutual coupling between antenna array elements. High-fidelity simulations using radar software (RFView~\cite{RFView}) further demonstrate that Brennan’s rule may be unreliable in complex scenarios~\cite{guerci2016new} where \gls*{icm} is prevalent.

Prior knowledge of radar data distribution is beneficial for modeling the clutter covariance matrix. Given the correlation between the neighboring pulses, a Toeplitz structure can be used to represent this correlation effectively \cite{Miller1987, Fuhramm1991}. Methods for estimating Toeplitz matrices were developed in \cite{Hongbin1999}. In \cite{Kang2015}, a low-rank constraint was further introduced to refine the clutter covariance model. However, the Toeplitz structure is only suitable for modeling correlations between temporally proximate pulses; or to account for the spatial correlation across a uniform line antenna array with equally spaced elements. Pulses that are farther apart in time may exhibit different returns, necessitating a more general banded structure for an accurate representation.

Tapered covariance matrices gradually reduce off-diagonal correlations, avoiding the sharp cutoff characteristic of the strictly banded matrices. The approach applies a tapering function \cite{WENDLAND1998258} to smoothly transition distant elements toward zero, preserving short-range correlations where relevant. The \gls*{tobasco} algorithm in \cite{ollila2021tobasco} provides a framework for covariance tapering which is useful in high-dimensional settings with short range dependencies. A survey of different types of tapers is included in \cite{ollila2021tobasco}. Banded covariance matrix models \cite{bickel2008regularization, rothman2010sparse} capture localized dependencies in high-dimensional radar data \cite{Mailloux1995, Guerci1999, Zatman2000}. The paper \cite{bein2015convex} develops a convex optimization framework for estimating banded covariance matrices, which offers a structured approach to reduce estimation error by exploiting local dependencies.
\subsection*{Contribution and Organization}

\begin{enumerate}
    \item In Sec.\,\ref{Sec:Model}, we propose the ``Banded+Spiked" structure for the clutter--plus--noise covariance matrix. 
    \item We demonstrate that the convex relaxation for maximizing the likelihood function for ``Banded+Spiked'' matrices can be formulated as minimization of a Frobenius norm in Theorem~\ref{Thm:UpperBound}. The proposed convex relaxation forms a provable upper bound for the likelihood function. 
    \item The convex relaxation is for both cases: when the number of samples is less than the dimension of the matrix and when the number of samples exceeds the matrix dimension.
    \item We propose a variational inequality bound for the tightness of the proposed convex relaxation in Theorem~\ref{Thm:EpsilonDerive}.
    \item In Sec.\,\ref{Sec:Algo}, we propose Algorithm \ref{alg:hat}, which jointly estimates noise power and a banded covariance matrix.
    \item We establish the conditions for the positive definiteness in Theorem~\ref{Thm:SigmaPositive} and the bandsize recovery of the estimated clutter--plus--noise covariance matrix. 
    \item In Sec.\,\ref{Sec:Num}, we show that the proposed algorithm yields a higher \gls*{scnr} compared to \gls*{tobasco} in \cite{ollila2021tobasco}; our previous work \cite{Jain2023}, which we call Spiked Covariance algorithm; and Diagonal loading  \cite{abramovich1981analysis} for an RFView\footnote{RFView is a high-fidelity radar simulation software which models complex airborne radar scenarios.} simulated environment.
\end{enumerate}

This paper extends the convex relaxation proposed in the conference paper \cite{Jain2025bandedMLE} to the regime where the covariance matrix dimension is larger than the number of available samples. Additionally, we derive bounds based on variational inequalities to quantify the quality of the relaxation and establish conditions that guarantee both the positive definiteness of the estimated clutter-plus-noise covariance matrix and the recovery of its bandsize.

In summary, we propose a convex relaxation, formulated as a Frobenius norm minimization, for the non-convex log-likelihood optimization of the ``Banded+Spiked" covariance matrix. This relaxation provides a provable upper bound for the likelihood maximization problem. Our method integrates the banded estimation framework from \cite{bein2015convex} to model the banded radar clutter sparsity under \gls*{icm} and incorporates \cite{Jain2023} for efficient noise power estimation. Unlike \cite{Jain2023}, which does not account for structured sparsity, our approach effectively captures both the banded structure and noise characteristics essential for radar applications.

\section{Convex Relaxation for ``Banded+Spiked" \gls*{mle}}
\label{Sec:Model}
In Sec.\,\ref{Sec:Model}-\ref{Sec:ModelA}, we first introduce the ``Banded+Spiked" structure to model the clutter--plus--noise covariance matrix. In Sec.\,\ref{Sec:Model}-\ref{Sec:ModelB}, we propose a convex relaxation for the negative log-likelihood function for estimating the ``Banded+Spiked" covariance matrix in Theorem~\ref{Thm:UpperBound}. We derive the convex relaxation for both cases: when the number of samples is less than the dimension of the matrix and when the number of samples exceeds the matrix dimension.
We outline the constraints that enforce the ``Banded+Spiked" structure, capturing both the localized dependencies and noise power dominated eigenvalues of the clutter--plus--noise covariance matrix. In Sec.\,\ref{Sec:Model}-\ref{Sec:VariationalBound} we state the bound of the norm difference between the argmins of the convex relaxation and non-convex negative log-likelihood function using the variational inequality bound in Theorem~\ref{Thm:EpsilonDerive}.
\subsection{``Banded+Spiked'' Covariance Model}
\label{Sec:ModelA}
The STAP datacube is a tensor of size $N_c \times N_s \times N_r$. Here, $N_c$ represents the number of channels, corresponding to the angular component of the datacube. $N_s$ denotes the number of pulses, each with the same waveform, capturing the slow time-scale or Doppler domain of the datacube. Lastly, $N_r$ denotes the number of range bins, which corresponds to the fast-time dimension within the datacube.

We assume that the waveform $\mathbf{z}\in \Complex^{P}$ is fixed for all pulses. For a given range gate $k_r$ and a single channel, we have the following equation for a given received pulse $x^{l}(k_r)$ indexed by $l$ and pulse $z(k)$ and noise $n(k)$, $1\leq k \leq P$:
\begin{equation}
\label{eq:Conv}
    x^{l}(k_r)=(\mathbf{h}_c^{l}\circledast \mathbf{z})(k_r)+n^{l}(k_r),
\end{equation}
where $\mathbf{h}_c^{l}\circledast \mathbf{z}$ is the convolution of the clutter impulse response $\mathbf{h}_c^{l}\in \Complex^{N}$ with the waveform $\mathbf{z}$, evaluated at the $k_r^{\text{th}}$ range bin. The receiver noise $n^{l}(k_r)\sim\mathcal{N}(0,\sigma^2)$ with variance $\sigma^2$ is evaluated at $k_r$ and is assumed to be independent for all elements in the data cube. We use the notation $\mathbf{h}_{c}^{l,k_r}$ to represent the elements of the vector $\mathbf{h}_{c}^{l}$ such that $\langle \mathbf{h}_{c}^{l,k_r}, \mathbf{z}\rangle$=$(\mathbf{h}_c^{l}\circledast \mathbf{z})(k_r)$, where $\langle\cdot,\,\cdot\rangle$ is the inner product. We stack all the responses in slow time to obtain a vector $\mathbf{x}\in\Complex^{N_s}$ for a given channel and a range bin $k_r$,
\begin{equation*}
    \mathbf{x}=\underbrace{\begin{bmatrix}
        \dots&\mathbf{h}_{c}^{1,k_r}&\dots\\
        \dots&\mathbf{h}_{c}^{2,k_r}&\dots\\
        \vdots&\vdots&\vdots\\
        \dots&\mathbf{h}_{c}^{N_s,k_r}&\dots
    \end{bmatrix}}_{\mathbf{H}_{k_r}}\mathbf{z}+\mathbf{n},
\end{equation*}
which can be represented as
\begin{equation}
    \mathbf{x}=\underbrace{(\mathbf{z}^{T}\otimes\mathbf{I})}_{\mathbf{Z}}\VEC(\mathbf{H}_{k_r})+\mathbf{n},
\end{equation}
where $\otimes$ denotes the Kronecker product, and $\VEC(\cdot)$ vectorizes the matrix one row after another. The covariance matrix of $\mathbf{x}$ is:
\begin{equation}
    \boldsymbol{\Sigma}=\mathbf{Z}\,\E[\VEC(\mathbf{H_{k_r}})\VEC(\mathbf{H}_{k_r})^H]\,\mathbf{Z}^{H}+\sigma^2\mathbf{I}.
\end{equation}
We assume that $\mathbf{h}_{c}^{i}$ and $\mathbf{h}_{c}^{j}$ are weakly 
 correlated, i.e., $\|\E[\mathbf{h}_c^{i}{\mathbf{h}_c^{j}}^{H}]\|_{\mathrm{op}}=\|\mathbf{R}_{c}^{ij}\|_{\mathrm{op}}< M_{t},\, |i-j|>L_{t},\,1\leq i,j\leq N_s $, where $\|\cdot\|_{\mathrm{op}}$ is the matrix operator norm.  $M_t$ is significantly smaller than the elements of the covariance matrix within the bandsize $L_{t}$. This condition is encountered when \gls*{icm} is prevalent as documented in \cite{ward1998space}. We also assume that $\E[\mathbf{h}_c^{l}]=\mathbf{0}$, for all pulses. This assumption imposes a banded structure over the temporal clutter--plus--noise covariance matrix which is defined as:
\begin{equation*}
   \boldsymbol{\Sigma}_{N_s\times N_s}=\mathbf{B}_{11}+\sigma^2\mathbf{I},
\end{equation*}
where $\mathbf{B}_{11}$ is a banded clutter covariance matrix for a single channel.
\begin{comment}
\begin{figure*}[h!]
\begin{align}
\label{eq:BandedDef}
    \boldsymbol{\Sigma}_{N_s\times N_s}=\underbrace{\mathbf{Z}\begin{bmatrix}
        \mathbf{R}_c^{11}&\mathbf{R}_c^{12}&\mathbf{0}_{k\times k}&\mathbf{0}_{k\times k}&\cdots&\mathbf{0}_{k\times k}\\
        \mathbf{R}_{c}^{21}&\mathbf{R_{c}^{22}}&\mathbf{R}_{c}^{23}&\mathbf{0}_{k\times k}&\cdots&\mathbf{0}_{k \times k}\\
        \mathbf{0}_{k\times k}&\mathbf{R}_{c}^{32}&\mathbf{R}_{c}^{33}&\mathbf{R}_{c}^{34}&\cdots&\mathbf{0}_{k\times k}\\
        \vdots&\vdots&\vdots&\vdots&\vdots&\vdots\\
        \mathbf{0}_{k\times k}&\cdots&\cdots&\cdots&\mathbf{R}_{c}^{N_s N_s-1}&\mathbf{R}_{c}^{N_s N_s}
    \end{bmatrix}\mathbf{Z}^{H}}_{\mathbf{B}}+\sigma^2\mathbf{I}.\\
\hline\nonumber
\end{align}
\end{figure*}
Note that in \eqref{eq:BandedDef}, $\mathbf{R}_{c}^{ij}={\mathbf{R}_{c}^{ji}}^{H}$ where $\E[\mathbf{h}_c^{i}{\mathbf{h}_c^{j}}^{H}]=\mathbf{R}_{c}^{ij}$.
\end{comment}
We further assume that nearby channels have weak spatial correlations when $1\leq q \leq N_c$ channels are concatenated, resulting in a covariance matrix of dimension $\matrixdimension = q N_s$.

\begin{equation*}
    \Cov_{\matrixdimension\times\matrixdimension}=\underbrace{\begin{bmatrix}
        \banded_{11}&\banded_{12}&\cdots&\banded_{1q}\\
        \banded_{12}&\banded_{22}&\cdots&\banded_{2q}\\
        \vdots&\cdots&\ddots&\vdots\\
        \banded_{q1}&\banded_{q2}&\cdots&\banded_{qq}
    \end{bmatrix}}_{\banded}+\sigma^{2}\I_{\matrixdimension\times\matrixdimension}
\end{equation*}
Under weak spatial correlations the $\|\banded_{ij}\|_{\mathrm{op}}< M,\,|i-j|>L$, where $M$ is significantly smaller in magnitude compared to the elements of covariance matrix within the bandsize $L$, imposing a banded structure. Furthermore, we assume that only a small fraction of the power of the clutter return signal lies above the noise floor, enforcing a spiked covariance model for the clutter-plus-noise covariance matrix. This assumption allows us to model the matrix with a few dominant eigenvalues representing clutter components, while the majority of the eigenvalues reflect the noise floor. The clutter--plus--noise covariance matrix is expressed as:
\begin{equation}
\label{eq:BandedplusEye}
    \Cov=\banded+\sigma^2\I,
\end{equation}
where $\banded$ is the banded spatio-temporal clutter covariance matrix and $\sigma^2$ is the noise power. RFview simulations in Fig.\,\ref{fig:BandedplusSpiked} motivate us to use the ``Banded+Spiked'' covariance matrix assumption.\footnote{Ideally, when concatenating multiple channels, a Block Banded model is more appropriate. However, if the inter-channel correlations are negligibly small, the banded assumption remains a reasonable approximation. We leave the exploration of the Block Banded model for future work.}.
\begin{figure}[htbp]
    \centering
    
    % Subfigure 1
    \begin{subfigure}{0.5\textwidth}
        \includegraphics[width=\linewidth]{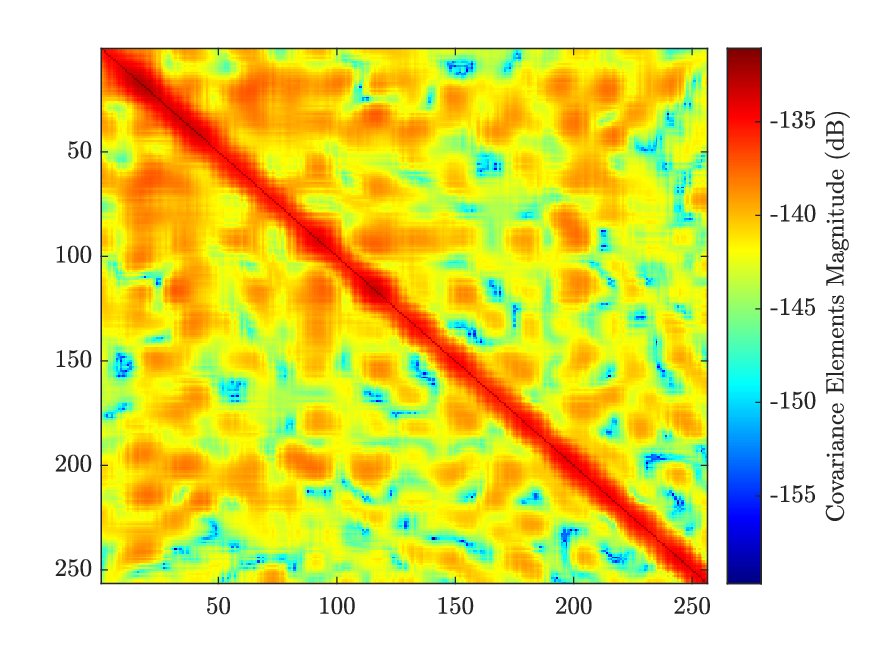}
        \caption{The clutter--plus--noise covariance matrix exhibits elements of higher magnitude concentrated near the diagonal, indicating its suitability for approximation using a banded matrix structure.}
        \label{fig:RFViewDoppler}
    \end{subfigure}
    \hfill
    % Subfigure 2
    \begin{subfigure}{0.5\textwidth}
        \includegraphics[width=\linewidth]{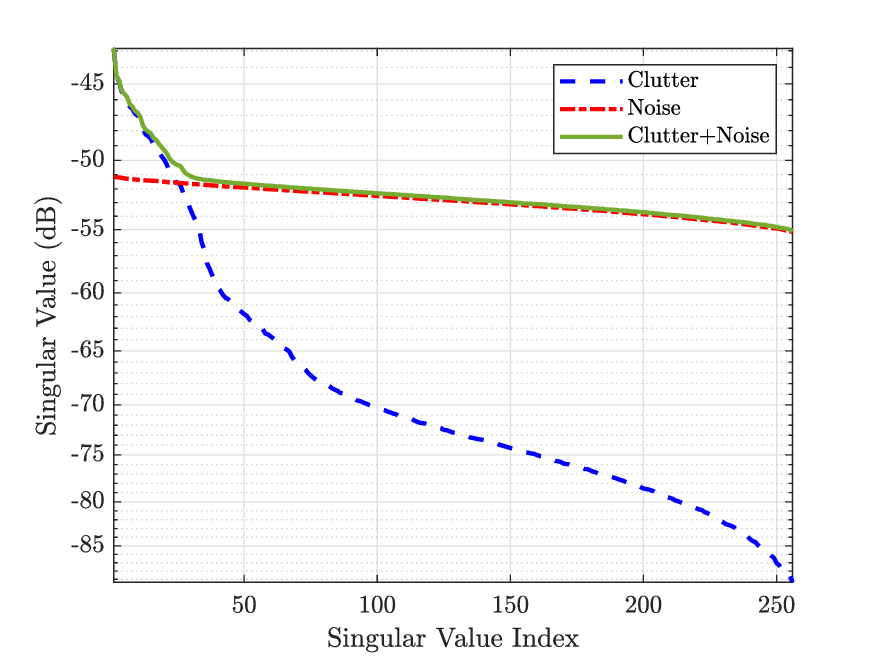}
        \caption{Only a small fraction of the singular values of the clutter covariance matrix (blue curve) exceed the noise floor (red curve). The singular values of the clutter--plus--noise covariance matrix (green curve) exhibits a spiked covariance model.}
        \label{fig:RFViewAzimuth}
    \end{subfigure}

    \caption{Clutter--plus--noise covariance matrix for $p = 256$ ($q = 4$ channels at $N_{s} = 64$ pulses/channel) simulated using RFView. }
    \label{fig:BandedplusSpiked}
\end{figure}
\subsection{Convex Relaxation for the Log-Likelihood Function}
\label{Sec:ModelB}
To estimate the clutter plus noise covariance matrix we use \gls*{mle}. We assume that the receive signal $\mathbf{x}\sim\mathcal{N}(\mathbf{0},\boldsymbol{\Sigma}_{\matrixdimension\times\matrixdimension})$. We use $K$ i.i.d. copies of $\mathbf{x}$ to compute the \gls*{scm},
\begin{equation}
\label{eq:SCM}
    \mathbf{S}=\frac{1}{K}\sum_{i=1}^{K}\mathbf{x}_i\mathbf{x}_i^{H}.
\end{equation}
The \gls*{scm} follows the complex Wishart distribution with $K$ degrees of freedom, when $p\leq K$:
\begin{equation}
    \label{eq:MLE}
    f(\SampCov; \Cov)= \frac{\det(\SampCov)^{K-\matrixdimension}}{\det(\Cov)^{K}\mathcal{C}\tilde{\Gamma}_{\matrixdimension}(K)}\exp(-K\tr(\Cov^{-1}\mathbf{S})),
\end{equation}
for a given $\Cov$, $\mathcal{C}\tilde{\Gamma}_{\matrixdimension}(K)$ is the complex multivariate gamma function. When $\matrixdimension>K$, \gls*{scm} follows the Complex Singular Wishart distribution \cite{RATNARAJAH2005399} given by,
\begin{equation}
    \label{eq:MLESingularWishart}
    f_{\mathrm{S}}(\SampCov;\Cov)=\frac{\pi^{K(K-\matrixdimension)}\det(\boldsymbol{\Lambda})^{K-\matrixdimension}}{\det(\Cov)^{K}\mathcal{C}\tilde{\Gamma}_{K}(K)}\exp(-K\tr(\Cov^{-1}\SampCov)),
\end{equation}
where $\boldsymbol{\Lambda}=\mathrm{diag}(\lambda_{1}(\SampCov),\,\cdots,\,\lambda_{K}(\SampCov))$. The negative log-likelihood functions for the case $\matrixdimension>K$ and $\matrixdimension\leq K$ are denoted as $F_{\mathrm{S}}(\Cov;\SampCov)$ and $F(\Cov;\SampCov)$, respectively and are expressed as:
\begin{equation}
\label{eq:LogLikelihood}
F(\SampCov;\,\Cov)=\tr(\Cov^{-1}\SampCov)+\log\det(\Cov\SampCov^{-1})-\matrixdimension+\mathrm{const.,}
\end{equation}
and
\begin{equation}
\label{eq:LogLikelihoodSing}
F_{\mathrm{S}}(\SampCov;\,\Cov)=\tr(\Cov^{-1}\SampCov)+\log\det(\Cov)-\log\det(\boldsymbol{\Lambda})+\mathrm{const.}
\end{equation}
where $\SampCov$, $\Cov$ and $\boldsymbol{\Lambda}$ are defined in \eqref{eq:MLE} and \eqref{eq:MLESingularWishart}.
The optimization problem for estimating the clutter covariance matrix involves minimizing the negative log-likelihood functions in \eqref{eq:LogLikelihood} and \eqref{eq:LogLikelihoodSing}, with respect to $\Cov$, where $\Cov=\mathbf{B}+\sigma^2\I$, and $\mathbf{B}$ is the banded clutter covariance matrix. %The objective function can be equivalently written as the Stein loss function given by, 
%$\mathcal F(\SampCov; \Cov)=\tr(\Cov^{-1}\SampCov)-\log\det(\Cov^{-1}\SampCov)-\matrixdimension,$
 The optimization problem is as follows:
\begin{align}
    \label{eq:Optim}
    \text{\textbf{Non-convex:}} &\quad \underset{\Cov\in \Complex^{\matrixdimension\times \matrixdimension}}{\text{Minimize}} \quad F(\SampCov; \Cov)\,\mathrm{or}\,F_{\mathrm{S}}(\SampCov; \Cov) \\
    &\quad\text{s.t.} \quad \Cov = \mathbf{B} + \sigma^2\I \nonumber \\
    &\quad\quad \sigma^2 \geq 0 \nonumber \\
    &\quad\quad \mathbf{B} \in \mathcal{B}_b \cap \mathcal{S}. \nonumber
\end{align}
Here, $\SampCov$ is the \gls*{scm}, and $\banded$ is the banded clutter covariance matrix and the set $\mathcal{B}_{b} := \{\banded : \banded = \banded^{H}, \banded \succ \mathbf{0}, \exists\,L \,\,\text{s.t.}\,\, |\banded_{jk}| < M,\,\,\forall\,\, |j-k| \geq L,\, 1 \leq j,k \leq \matrixdimension\}$ is a convex set with interior point \cite{bein2015convex}. The value $M$ is significantly smaller than the elements of the covariance matrix inside the bandsize $L$. The set does not impose a hard zero constraint on any element outside the band size \( L \). The set $\mathcal{S} := \{\mathbf{M} : \mathbf{M} = \mathbf{M}^{H}, \mathbf{M} \succeq \mathbf{0}, \exists\, r \ll \matrixdimension \,\,\text{s.t.}\,\, \lambda_i(\mathbf{M}) \ll \sigma^2,\, r+1 \leq i \leq \matrixdimension\}$ represents the set of spiked covariance matrices, where only $r$ eigenvalues $\lambda_i(\mathbf{M})$ exceed the noise power $\sigma^2$. The parameter $r$ reflects the index where the eigenvalues of the clutter covariance matrix are higher than the noise floor. We assume that the parameters $L$ and $r$ are unknown. 

The objective function in \eqref{eq:Optim} is non-convex with respect to $\Cov$. Moreover, while the set $\mathcal{B}_{b}$ is convex with respect to $\Cov$, it becomes non-convex with respect to $\Cov^{-1}$. On the other hand, the set of spiked covariance matrices $\mathcal{S}$ is non-convex with respect to both $\Cov$ and $\Cov^{-1}$. Typically, enforcing the constraints in $\mathcal{S}$ requires a convex relaxation. However, if the noise power $\sigma^2$ is known beforehand, the structure $\banded + \sigma^2\I$ inherently satisfies the spiked covariance property, allowing us to bypass the need for convex relaxation for the set $\mathcal{S}$ as done in \cite{donoho2018optimal}. When the noise power is known, both the functions in \eqref{eq:LogLikelihood} and \eqref{eq:LogLikelihoodSing} can be upper bounded by the Frobenius norm as stated in the following theorem under the assumption that largest eigenvalues for $\Cov$ is some finite value\footnote{This assumption is practical as the likelihood function is only defined when the eigenvalues of $\Cov$ are finite}. 
\begin{assumption}
\label{Assumption:Compact}
    For the clutter--plus--noise covariance matrix of dimension $p$, $\Cov$, the noise power $\sigma^2$ is known and $\lambda_{\min}(\Cov)=\sigma^2$. Additionally, there exists a positive constant $c$ such that  $\lambda_{\max}(\Cov)\leq c$.
\end{assumption}
\begin{theorem}
\label{Thm:UpperBound}
    Assume that (A\ref{Assumption:Compact}) holds, then for an \gls*{scm} $\SampCov$, the following convex upper bounds satisfy: 
    \begin{equation}
    \label{eq:UpperHigh}
        F(\Cov;\,\SampCov)\leq \left\|(\Cov-\SampCov)/\sigma^2\right\|_{F}^2+\mathrm{const.},
    \end{equation}
    where $F(\Cov;\,\SampCov)$ is defined in \eqref{eq:LogLikelihood}, and 
    \begin{align}
    \label{eq:UpperLow}
        F_{\mathrm{S}}(\Cov;\,\SampCov)\leq& \sqrt{\matrixdimension}\left\|(\Cov-\SampCov)/\sigma^2\right\|_{F}+\matrixdimension\nonumber\\&+\matrixdimension\log\,c-\log\det\boldsymbol{\Lambda}+\mathrm{const.}\,,
    \end{align}
    where $F_{\mathrm{S}}(\Cov;\,\SampCov)$ is defined in \eqref{eq:LogLikelihoodSing}.
\end{theorem}
\noindent The proof is in the Appendix.\hfill \(\blacksquare\)

\noindent We define the upper-bound for the case $\matrixdimension\leq K$ as:
\begin{equation}
\label{eq:UpperBoundHigh}
    \upperbound(\Cov;\;\SampCov)=\left\|(\Cov-\SampCov)/\sigma^2\right\|_{F}^2+\const\,.
\end{equation}
For the case $\matrixdimension>K$ the definition is:
\begin{align}
    \label{eq:UpperBoundLow}
    \upperbound_{S}(\Cov;\;\SampCov)&=\sqrt{\matrixdimension}\left\|(\Cov-\SampCov)/\sigma^2\right\|_{F}+\matrixdimension\nonumber\\&+\matrixdimension\log\,c-\log\det\boldsymbol{\Lambda}+\const\,.
\end{align}
Thus, the minimizer of the squared Frobenius norm provides an estimate for the clutter-plus-noise covariance matrix. We apply results from \cite{bein2015convex} to solve the following convex relaxation problem, assuming we have separately estimated the noise power $\hat{\sigma}^2$:
\begin{align}
\label{eq:OptimRelax}
    \text{\textbf{Convex Relaxation:}} \quad 
    \underset{\Cov \in {\Complex}^{\matrixdimension \times \matrixdimension}}{\text{Minimize}} & \quad \left\| \Cov - \SampCov \right\|_{F}^{2} \\
    \text{s.t.} & \quad \Cov = \banded + \hat{\sigma}^2 \I \nonumber \\
    & \quad \banded \in \mathcal{B}_{b}, \nonumber
\end{align}
%The matrix $\SampCov$ is the \gls{scm} and $\mathcal{B}$ is defined in \eqref{eq:Optim}. 
where $\mathcal{B}_{b}$ and $\SampCov$ are defined in \eqref{eq:Optim}. 

\subsection{A Variational Inequality based Bound for the Convex Relaxation}
\label{Sec:VariationalBound}
%One way to assess the quality of a bound is by finding a concave lower bound \( \lowerbound_{\mathrm{conc}} \) and a convex upper bound \( \upperbound_{\mathrm{conv}} \), such that \( \lowerbound_{\mathrm{conc}} \leq F \leq \upperbound_{\mathrm{conv}} \). In this case, if \( x^{*} = \arg\min F(x) \), \( \underline{x} = \arg\max \lowerbound_{\mathrm{conc}}(x) \), and \( \bar{x} = \arg\min \upperbound_{\mathrm{conv}}(x) \), then the function values satisfy \( \lowerbound_{\mathrm{conc}}(\underline{x}) \leq F(x^{*}) \leq \upperbound_{\mathrm{conv}}(\bar{x}) \). However, this approach does not provide any guarantees on the closeness of \( x^{*} \) to either \( \bar{x} \) or \( \underline{x} \), where closeness refers to the norm difference between these values. One possible approach is to derive the \gls*{crb} for the optimization problem in \eqref{eq:Optim}, but this is challenging due to the presence of non-smooth constraints. Instead, we employ the variational inequality approach from \cite{pflug2012optimization}, which quantifies how far the solution obtained from the convex relaxation deviates from the solution of the original non-convex problem.

Since the convex relaxation \(F_\conv\)  does not converge uniformly to the non-convex function $F$,  it follows that \(\argmin\,F_\conv\) is not necessarily close to \(\argmin\,F\).  Therefore, in this section we construct a variational inequality~\cite{pflug2012optimization} to bound the distance between \(\argmin\,F_\conv\) and \(\argmin\,F\).  Although, this bound is loose, it still provides a useful ballpark figure.

\begin{theorem}[\hspace{-0.2mm}\cite{pflug2012optimization}, Lemma~1.9]
\label{Thm:ConvRelaxation}
    Let $\upperbound(\SampCov;\;\Cov)$ be a convex relaxation for the function $F(\SampCov;\;\Cov)$. Suppose there exists an $\epsilon$ such that 
    \[\sup_{\Cov}|\upperbound(\SampCov;\;\Cov)-F(\SampCov;\;\Cov)|\leq \epsilon,\]
    and \[\upperbound(\SampCov;\;\Cov)\geq \underset{\Cov}{\mathrm{argmin}}\,\upperbound(\SampCov;\;\Cov)+m\left\|\Cov-\underset{\Cov}{\mathrm{argmin}}\,\upperbound(\SampCov;\;\Cov)\right\|_{F}^\gamma,\]
    for some $m$ and $\gamma$.
    Then,
    \[\left\|\underset{\Cov}{\mathrm{argmin}}\,F(\SampCov;\;\Cov)-\underset{\Cov}{\mathrm{argmin}}\,\upperbound(\SampCov;\;\Cov)\right\|_F\leq \left(\frac{2\epsilon}{m}\right)^{\frac{1}{\gamma}}.\]
    
\end{theorem}
\noindent Since, in our case the upper bound is the Frobenius norm, if we chose $m=1$, then the growth condition is satisfied. We need to show that sup-norm of $\upperbound(\Cov)-F(\Cov)$ is bounded, which is shown in the following theorem.
\begin{theorem}
\label{Thm:EpsilonDerive}
    Assume that (A\ref{Assumption:Compact}) holds, then $\epsilon$ and $\gamma$ in Theorem~\ref{Thm:ConvRelaxation} for the given convex relaxation in Theorem~\ref{Thm:UpperBound} are,
    \begin{enumerate}
        \item Case $p\leq K$, with $F(\SampCov;\;\Cov)$ defined in \eqref{eq:LogLikelihood} and $\upperbound(\SampCov;\;\Cov)$ in \eqref{eq:UpperBoundHigh}:
        \begin{align}
            \epsilon=&\frac{\matrixdimension}{\sigma^4}(c+\lambda_{\max}(\SampCov))^2+\matrixdimension\log\Big(\frac{\lambda_{\max}(\SampCov)}{\sigma^2}\Big)\nonumber\\ &+\matrixdimension-\frac{\tr(\SampCov)}{c}.
        \end{align}
        For this case, $\gamma=2$.
    \item Case $p>K$, with $F(\SampCov;\;\Cov)$ defined in \eqref{eq:LogLikelihoodSing} and $\upperbound(\SampCov;\;\Cov)$ in \eqref{eq:UpperBoundLow}:
            \begin{align}
            \epsilon=\frac{\matrixdimension}{\sigma^2}(c+\lambda_{\max}(\SampCov))+\matrixdimension\log\frac{c}{\sigma^{2}}+\matrixdimension-\frac{\tr(\SampCov)}{c}.
        \end{align}
        For this case, $\gamma=1$.
    \end{enumerate}
    %The matrix $\SampCov$ is the \gls*{scm} constructed using $K$ samples.

\end{theorem}
\noindent The proof is in the Appendix.\hfill \(\blacksquare\)

\noindent For the case \(\matrixdimension \leq K\), the term \((2\epsilon)^{1/2}\), which bounds the distance between the minimizing arguments of the non-convex objective and its convex relaxation, scales as \(\mathcal{O}(\sqrt{p})\). For the case \(p > K\), the bound becomes $2\epsilon$ and follows a growth rate of \(\mathcal{O}(\matrixdimension)\). This change in scaling behavior arises due to differences in the convex relaxation function for each regime. Specifically, when \(\matrixdimension \leq K\), the convex relaxation corresponds to the squared Frobenius norm, whereas for \(\matrixdimension > K\), it is the Frobenius norm itself. Despite this structural difference, the optimization problem remains unchanged in both cases, as the optimal solution is the same whether considering the squared or unsquared Frobenius norm. The only distinction lies in the derived bounds.

The assumption (A\ref{Assumption:Compact}) ensures that the search space is a compact set. Bounds for both cases grow linearly with respect to the constant $c$. We use the following function to quantify the normalized behavior of the bound. We define
\[\mathrm{Relative~Bound}:=\Big(\frac{2\epsilon}{p}\Big)^{\frac{1}{\gamma}}\cdot\frac{1}{c}.\]
The above relative bound captures how good the bound is normalized over the search space represented using the constant $c$ used in (A\ref{Assumption:Compact}). In Fig.\,\ref{fig:epsilonPlot}, we plot the Relative Bound with respect to \(c\) for both regimes for different values of $K$. We assume \(\sigma^2 = 1\), allowing us to represent \(c\) as a ratio, expressed in dB. For \(\matrixdimension=256\), we construct a synthetic matrix with $25$ spikes and bandsize $L=4$, generate the corresponding \gls*{scm}, and compute \(\lambda_{\max}(\SampCov)\) and \(\tr(\SampCov)\). The results are averaged over 100 Monte Carlo simulations to obtain reliable estimates. We observe that as the search space is increased the relative bound decreases. The limiting value of the bound will be $1$ as $p,c\rightarrow\infty$. 

\noindent\textbf{Remark:} The derived variational bounds are loose, as noted in \cite{pflug2012optimization}, because the upper bounding function is not globally Lipschitz. If we introduce an additional constraint in the optimization problem \eqref{eq:Optim}—specifically, that the \( \lambda_{\max}(\Cov) \) is bounded by a fixed constant—then the spiked covariance structure may no longer hold. Finding a tighter bound that is globally Lipschitz continuous and aligns with standard objective functions used in statistical learning theory for matrix estimation remains an open problem, which we leave for future work.
 
We formulated the model for the clutter--plus--noise covariance matrix and proposed an optimization problem to estimate the matrix. We used variational inequality based bound to assess the quality of the convex relaxation. In the next section, we will explain the estimation of the noise power $\sigma^2$ using the technique in \cite{Jain2023} and the algorithm for solving \eqref{eq:OptimRelax}.

\begin{figure}[htbp]

        \includegraphics[width=\linewidth]{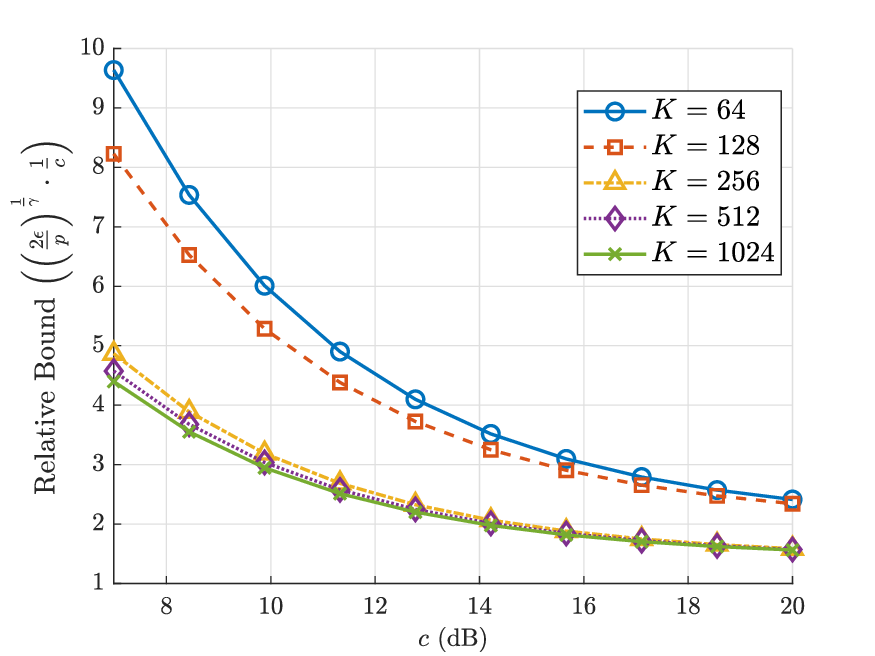}
    \caption{The relative bound decreases as the number of samples \( K \) increases. This is because the sample covariance matrix \( \SampCov \) becomes a more accurate approximation of the true covariance matrix from above. Additionally, as \( c \) increases, the bound further decreases. While the numerical values on the \( y \)-axis have limited quantitative significance, the plot qualitatively illustrates that the bound improves as \( K \to \infty \).
 }
    \label{fig:epsilonPlot}
\end{figure}
\begin{comment}
\textbf{Case 1:} $K\geq \matrixdimension$, in this case we have $\lowerbound(\Cov)=\tr(\Cov^{-1}\SampCov)+\tr(\Cov_{0}^{-1}(\Cov-\Cov_0))+\log\det(\SampCov^{-1}\Cov_{0})-\matrixdimension$ and $\bar{F}=\left\|(\Cov-\SampCov)/\sigma^2\right\|_{F}^{2}$. In this case $\lowerbound$ satisfies the following growth condition:

\[\lowerbound(\Cov)\geq \arg\min\lowerbound(\Cov)+\frac{1}{2}c\left\|\Cov-\arg\min\lowerbound(\Cov)\right\|_{F}^2\]
where . Now given that $K>\matrixdimension$, the value of $\lambda_{min}(\SampCov)=0$ only for a set of measure zero, therefore $\lowerbound$ is strongly convex almost everywhere in that case we have $\gamma=2$.

\textbf{Case 2:} $K<\matrixdimension$, in this case we do not have a lower bound with quadratic growth rate, in this case we go with the sublinear growth rate. We use the Complex Singular Wishart distrbution and the lower bound with respect to $\Cov$. 
\[\lowerbound=\tr(\Cov^{-1}\SampCov)+\mathrm{constant}.\]
\end{comment}
\section{Covariance Estimation Algorithm}
\label{Sec:Algo}
The joint estimation of noise power and banded covariance is presented in Algorithm~\ref{alg:hat}, detailed in Sec.\,\ref{Sec:Algo}–\ref{Sec:NoiseEstPlusBanded}. In Sec.\,\ref{Sec:Algo}-\ref{Sec:Thms}, we state the conditions under which the recovered clutter--plus--noise covariance matrix is positive definite in Theorem~\ref{Thm:SigmaPositive} and state the recovery of the bandsize in Theorem~\ref{Thm:BandRecover}.
\subsection{Banded+Spiked Covariance Estimation}
\label{Sec:NoiseEstPlusBanded}
In Sec.\,\ref{Sec:Algo}-\ref{Sec:NoiseEstPlusBanded}\ref{Sec:NoiseEst} we propose a technique for noise power estimation. In Sec.\,\ref{Sec:Algo}-\ref{Sec:NoiseEstPlusBanded}\ref{Sec:BandedEst} we outline the method for estimating the banded covariance matrix. We conclude this subsection by presenting Algorithm~\ref{alg:hat}, which summarizes the complete procedure.

\subsubsection{Noise Power Estimation}
\label{Sec:NoiseEst}
The constraint on the parameter $r$ in the definition of $\mathcal{S}$ in \eqref{eq:Optim} is non-convex. However, by assuming a spiked covariance matrix structure, we can leverage the approach in \cite{donoho2018optimal, Jain2023} to estimate both the parameter $\hat{r}$ and the noise power $\hat{\sigma}^2$. Notably, we primarily need an accurate estimate of $\sigma^2$, for which consistency is established in \cite{donoho2018optimal}. The estimation procedure is as follows:
\begin{enumerate}
    \item Compute the singular values $\lambda$ of the \gls*{scm} $\SampCov$ in \eqref{eq:SCM} with $K$ samples, using the \gls*{svd} algorithm The singular value computation has a complexity of $\mathcal{O}(\matrixdimension^3)$. 
    \item Obtain the median of the singular values and normalize it by the median of the \gls*{mp} distribution with parameter $\matrixdimension/K$. This normalized value provides the estimate $\hat{\sigma}^2$. Assuming the \gls*{svd} algorithm outputs the singular values in sorted order, this step has a complexity of $\mathcal{O}(1)$; otherwise it is $\mathcal{O}(\matrixdimension\log\matrixdimension)$. 
\end{enumerate}
The estimate $\hat{r}$ is then determined by counting the number of singular values that exceed $\hat{\sigma}^2$. We adopt this approach due to its efficiency over the \gls*{mle} based methods \cite{kang2014rank} which incur a computational overhead of $\mathcal{O}(\matrixdimension^2)$ to estimate the noise power. In contrast, our method achieves this in $\mathcal{O}(1)$ or in the unsorted case $\mathcal{O}(\matrixdimension\log\matrixdimension)$, given that the \gls*{svd} algorithm returns the singular values in sorted order. The median for \gls*{mp} distribution needs to be computed numerically for the given ratio $\matrixdimension/K$, which can be done offline depending upon the desired accuracy.
\vspace{-2mm}
\subsubsection{Banded Matrix Estimation}
\label{Sec:BandedEst}
We use the \gls*{bcd} based convex optimization algorithm used in \cite{bein2015convex} to the estimate of the clutter--plus--noise covariance matrix. Conventional convex optimization algorithms require that the objective function and the constraint should be smooth with respect to the optimization variables. However, enforcing a banded structure to a matrix introduces a non-smooth convex constraint where conventional convex optimization algorithms fail. Eq.\,\eqref{eq:OptimRelax}, can be expressed as the following optimization problem, once we have estimated the noise power $\hat\sigma^2$ from the previous section:
  \begin{align}
    \label{eq:primal}
    \hat\Cov = \arg\min_{\Cov}\left\{ \half\|\Cov-(\S+\hat\sigma^2\I)\|_F^2+\mu\|\Cov\|_{2,1}^* \right \}.
\end{align}
The parameter $\mu$ is the regularization parameter and the penalty function $\left\|\cdot\right\|_{2,1}^{*}$ is defined as follows:
\begin{align}
  \label{eq:penalty}
  \|\Cov\|_{2,1}^*=\sum_{\ell=1}^{\matrixdimension-1}\sqrt{\sum_{m=1}^\ell w_{\ell m}^2\left\|\VEC(\Cov_{s_m})\right\|_2^2}.
\end{align}
The set of indices in \eqref{eq:penalty},\,$s_m$, also called group lasso indices, enforce the banded sparsity and is defined as:
\begin{equation}
\label{eq:BandedSparse}
    s_m=\{(j,\,k):|j-k|=\matrixdimension-m\},
\end{equation}
and the weights $w_{\ell m}$ represent the generalized hierarchical penalty:
\begin{equation}\label{eq:genw}
w_{\ell,m} = \frac{\sqrt {2\ell}}{\ell - m +1},  \ \mbox{for}\  1 \leq m \leq \ell, \  1 \leq \ell \leq \matrixdimension - 1,
\end{equation}
  with $1\le m\le\ell\le \matrixdimension$.
With \eqref{eq:penalty}, we have modeled the bandedness defined in the set $\mathcal{B}_{b}$ in \eqref{eq:OptimRelax}. The index set $s_m$ defined in \eqref{eq:BandedSparse} and the weights $w_{\ell m}$ enforce the banded structure in matrix $\Cov$. The clutter covariance matrix $\banded=\Cov-\sigma^2\I$. Note that the indexing in $s_m$ is from corners of the matrix towards the diagonal. The penalty function can be expressed as a element-wise Hadamard product ($\odot$) with a sequence of weighting matrices $\{\mathbf{W}^{(\ell)}\}$,
\begin{equation}
    \label{eq:WeightMatrix}
    \mathbf{W}^{(\ell)}_{s_m}=w_{\ell m}\mathbbm{1}_{\{m\le\ell\}},
\end{equation}where $w_{\ell m}$ is defined in \eqref{eq:genw} and $\mathbbm{1}_{\{\cdot\}}$ is the indicator function. Eq.\eqref{eq:penalty} in terms of weighting matrix is expressed as follows:
\begin{align}
    \label{eq:groupLasso}
    &\sum_{\ell=1}^{\matrixdimension-1}\sqrt{\sum_{m=1}^\ell w_{\ell m}^2\left\|\VEC(\Cov_{s_m})\right\|_2^2}\\\nonumber &=\sum_{\ell=1}^{\matrixdimension-1}\left\|\VEC((\mathbf{W}^{(\ell)}\odot\Cov)_{g_{\ell}})\right\|_{2}.
\end{align}
The set \( g_{\ell} := \cup_{m=1}^{\ell}s_m \) is the union of the index sets \( s_m \) defined in \eqref{eq:BandedSparse}, which captures the sparsity pattern as we move from the corner of the matrix toward the diagonal. Eq.~\eqref{eq:groupLasso} represents a group lasso formulation, as it involves the summation of non-squared 2-norms, akin to minimizing the \(\ell_1\)-norm. Here, the matrix \( \mathbf{W}^{(\ell)} \odot \Cov \) is vectorized, and the elements corresponding to \( g_{\ell} \) are selected from \( \VEC(\mathbf{W}^{(\ell)} \odot \Cov) \), followed by taking their 2-norm. This procedure is applied for all \( 1 \leq \ell \leq \matrixdimension - 1 \), thereby enforcing a group lasso penalty.
 The indexing set enforces a group lasso penalty and is non-smooth but is convex.
It has been shown in \cite{bein2015convex} and the references therein that the dual of \eqref{eq:primal} is separable; therefore, \gls*{bcd}
can be used to solve the problem similar to the lasso algorithm. The dual of \eqref{eq:primal} is given as:
\begin{align}
\label{eq:dual}
    \underset{A^{(\ell)}\in\Complex^{\matrixdimension\times \matrixdimension}}{\text{Minimize}}&\quad \half
    \left\|\S+\hat\sigma^2\I-\mu\sum_{\ell=1}^{\matrixdimension-1}
    \mathbf{W}^{(\ell)}\odot\mathbf{{\Al}}\right\|_F^2\\
      \text{s.t.}&\quad \|\mathbf{\Al}_{\gl}\|_2\le 1,~\mathbf{\Al}_{\gl^c}=0,~~1\le\ell\le \matrixdimension-1.\nonumber
\end{align}
The set $g_\ell$ and the weight matrices, $\mathbf{W}^{(\ell)}$ are defined in \eqref{eq:WeightMatrix}. The formulation of the dual \eqref{eq:dual} from the primal \eqref{eq:primal} is derived in \cite[Appendix A]{bein2015convex}, we add the factor $\hat\sigma^2\I$ to enforce the spiked covariance condition.

Algorithm \ref{alg:hat} gives the joint estimates for the noise power $\hat\sigma^2$ and the \gls*{bcd} algorithm for solving
\eqref{eq:dual}, which by the primal-dual relation in turn gives a solution to \eqref{eq:primal}.
The matrices $\mathbf{A}^{(\ell)}$ correspond to each dual variable matrix. The update over each $\mathbf{\Al}$ involves projection onto an
ellipsoid, which amounts to finding a root of the univariate function,
\begin{align}
  h_\ell(\nu)=\sum_{m=1}^\ell\frac{w_{\ell m}^2}{(w_{\ell
      m}^2+\nu)^2}\|\mathbf{\hat R}^{(\ell)}_{s_m}\|^2-\mu^2,\label{eq:nu}
\end{align}
where $w_{\ell m}$ is defined in \eqref{eq:genw}, and $\hat{\mathbf{R}}$ is defined in Algorithm~\ref{alg:hat}. The term \( \max\{0, \hat{\nu}_{\ell} \} \) in the Step 2 of the inner iteration in Algorithm~\ref{alg:hat} is analogous to the power allocation process in the water-filling algorithm Here, power is allocated in terms of the \(\ell_2\)-norm to the subdiagonals of the estimated matrix \( \hat{\mathbf{A}} \). Moreover, the \gls*{kkt} conditions for \eqref{eq:dual} impose that \( \hat{\nu}_{\ell} \geq 0 \), ensuring that only non-negative roots are admissible for \eqref{eq:nu}. For a more in-depth analysis on ellipsoidal projection, see \cite[Appendix~B]{bein2015convex}.

The computational complexity of Algorithm~\ref{alg:hat} depends on the required accuracy for the root-finding process. Given that we perform $\matrixdimension$ matrix multiplications of size $\matrixdimension \times \matrixdimension$, the overall complexity is $\max(\mathcal{O}(\matrixdimension^4), \mathcal{O}(C_r))$, where $C_r$ denotes the complexity of the root-finding algorithm in \eqref{eq:nu}. The algorithm is computationally expensive compared to our previous work \cite{Jain2023}. Algorithm~\ref{alg:hat} is also suitable for any other tighter convex upper bound, as it can serve as a base problem within a proximal gradient algorithm~\cite{rockafellar1997convex, nesterov2013gradient} that leverages this refined bound.

\subsection{Positive Definiteness and Bandsize Recovery for the Estimated Covariance Matrix}
\label{Sec:Thms}
To establish the positive definiteness of the estimated covariance matrix \( \hat{\Cov} \), we utilize \cite[Theorem~10]{bein2015convex}. We begin by defining a random set and stating some assumptions used in \cite{bein2015convex} to facilitate the analysis. Then, using Theorem~\ref{Thm:psd} we show in Theorem~\ref{Thm:SigmaPositive} that \( \hat{\Cov} \) remains positive definite with high probability over the random set defined below.

\begin{definition}
\label{Def:RandomSet}
    For any $x>0$, the random set $\mathcal{A}_{x}$ is expressed as,
    \[\mathcal{A}_{x}:=\left\{\SampCov:\,\max_{1\leq i,j\leq \matrixdimension}|\SampCov_{ij}-\Cov^{*}_{ij}|\leq x\sqrt{\log\matrixdimension/n}\right\},\]
    where $\SampCov$ is the \gls*{scm} with $K$ samples and $\Cov^*$ is the covariance matrix of dimension $\matrixdimension$. 
\end{definition}
\begin{assumption}
    \label{Assumption:Tail}
    Let $\mathbf{x}=(x_1,\,x_2,\,\cdots,\,x_\matrixdimension)^{T}\in\Complex^{p\times 1}$ with zero mean, $\E[\mathbf{x}]=\mathbf{0}$, and denote $\E[\mathbf{x}\mathbf{x}^{H}]=\Cov^{*}$. Each $X_j$ is marginally sub-Gaussian:
\[\E\left[\exp\left(\frac{t x_{j}}{\sqrt{\Cov_{jj}^{*}}}\right)\right]\leq \exp(Ct^2),\]
for all $t\geq 0$ and for some constant $C>0$ that is independent of j. Moreover, $\max_{ij}|\Cov_{ij}^{*}|\leq S_0$, for some constant $S_0>0$.
\end{assumption}
\begin{assumption}
    \label{Assumption:Growth}
    The dimension $\matrixdimension$ can grow with $K$ at most exponentially in $K$: $\kappa_{0}\log K\leq \log\matrixdimension\leq \kappa_{1}K$, for some constant $\kappa_{0},\,\kappa_{1}>0$.
\end{assumption}
\begin{assumption}
\label{Assumption:LeastEig}
    The least eigenvalue of $\Cov^{*}$ satisfies \(\lambda_{\min}(\Cov^{*})\geq 2C'L\sqrt{\log\matrixdimension/K}\) with the bandsize $L$ for some constant $C'>0$.
\end{assumption}
\begin{assumption}
\label{Assumption:BandedSNR}
    The value \(\min_{\ell\in B(\Cov^{*})}\left\|\Cov_{s_{\ell}}^{*}\right\|_{2}/\sqrt{2\ell}\geq c'\) for some constant $c'>0$. The set $s_{\ell}$ is defined in \eqref{eq:BandedSparse}. The set $B(\Cov^{*})$ is an index such that $\|{\Cov^{*}}_{g_{B(\Cov^{*})}}\|_2\leq M_1$ and $\|{{\Cov^{*}}_{s_{B(\Cov^{*})+1}}}\|_2> M_2$, where $M_1,M_2\geq0$ are arbitrarily small constants and the set $g_{\ell}:=\cup_{m=1}^{\ell}s_{m}$.
\end{assumption}
\noindent We state Theorem~\ref{Thm:psd} when no spiked covariance matrix model is assumed, i.e., $\sigma^2=0$.
\begin{theorem}[\hspace{-0.2mm}\cite{bein2015convex},\,Theorem\,10]
\label{Thm:psd}
Suppose that (A\ref{Assumption:Tail}), (A\ref{Assumption:Growth}), (A\ref{Assumption:LeastEig}) and (A\ref{Assumption:BandedSNR}) hold. The estimator defined in \eqref{eq:primal} when $\hat\sigma^2=0$, with weights given by \eqref{eq:genw}, and with regularization parameter in \eqref{eq:primal}, \gls*{scm} $\SampCov$ with $K$ samples, $\mu=2x\sqrt{\log\matrixdimension/K}$, has minimum eigenvalue at least $C'L\sqrt{\log\matrixdimension/K}>0$ with probability larger than $1-c_1/\matrixdimension$, for some $c_1>0$ over the set $\mathcal{A}_{x}$ defined in Def.\,\ref{Def:RandomSet}.
\end{theorem}

The assumptions (A\ref{Assumption:Tail}) and (A\ref{Assumption:Growth}) used in Theorem~\ref{Thm:psd} are standard statistical assumptions commonly used in high-dimensional statistics \cite{bickel2008regularization} and are practical in a \gls*{stap} setting. Assumption (A\ref{Assumption:BandedSNR}) imposes a signal-to-noise ratio like constraint on the banded structure of the covariance matrix—specifically, it compares the magnitude of matrix elements within the bandsize $L$ to those outside it. Under this assumption, elements within the band are required to be larger in magnitude than those outside by a factor of $M_1/M_2$. This assumption directly relates to the parameter $M$ in the feasible set $\mathcal{B}_{b}$ defined in \eqref{eq:Optim}, with the ratio $M_1/M_2$ effectively placing an upper bound on the allowable value of $M$. Furthermore, based on  (A\ref{Assumption:LeastEig}) used in Theorem~\ref{Thm:psd}, we establish a condition on $\sigma^2$, ensuring that the estimated covariance matrix $\hat{\Cov}$ remains positive definite. While we have already estimated $\lambda_{\min}(\hat{\Cov}) = \hat{\sigma}^{2}$, solving the optimization problem \eqref{eq:primal} may result in $\hat{\Cov}$ becoming indefinite.

\begin{theorem}
\label{Thm:SigmaPositive}
Assume that (A\ref{Assumption:Tail}), (A\ref{Assumption:Growth}), and (A\ref{Assumption:BandedSNR}) hold. Suppose that the minimum eigenvalue of the clutter-plus-noise covariance matrix $\Cov$ with bandsize $L$ and dimension $\matrixdimension$ satisfies  
\[
\lambda_{\min}(\Cov) = \sigma^2 \geq C' L \sqrt{\frac{\log \matrixdimension}{K}},
\]  
where $K$ is the number of samples used to construct the \gls*{scm} $\SampCov$, for some constant $C'>0$. Then the estimated clutter-plus-noise covariance matrix $\hat{\Cov}$ has a minimum eigenvalue satisfying  \(\lambda_{\min}(\hat{\Cov}) \geq \hat{\sigma}^2\) with probability at least $1 - c_1 / \matrixdimension$, for some constant $c_1 > 0$, over the set $\mathcal{A}_x$ defined in Def.~\ref{Def:RandomSet}.
  
\end{theorem}
\noindent The proof is in the Appendix.\hfill \(\blacksquare\)

\noindent Theorem~\ref{Thm:psd} and Theorem~\ref{Thm:SigmaPositive} require some estimate of the the bandsize $\hat{L}$ which is stated in \cite[Theorem\,4]{bein2015convex}.
\begin{theorem}[\hspace{-0.2mm}\cite{bein2015convex},\,Theorem\,4]
\label{Thm:BandRecover}
Assume that (A\ref{Assumption:Tail}), (A\ref{Assumption:Growth}) and (A\ref{Assumption:BandedSNR}) hold. Then for a given $x>0$ such that the regularization parameter in \eqref{eq:primal}, $\mu\geq x\sqrt{\log\matrixdimension/K}$, then $\hat{L}=L$ on the set  $\mathcal{A}_{x}$ defined in Def.~\ref{Def:RandomSet} with probability at least $1-c_1/p$, for some constant $c_1>0$.
    
\end{theorem}
\noindent\textbf{Remark:} Theorems~\ref{Thm:psd}, \ref{Thm:SigmaPositive}, and \ref{Thm:BandRecover} provide only the theoretical foundation for the existence of various properties and parameter recovery under certain constants. However, in practice, these constants are unknown. A practical rule of thumb, as suggested in \cite{bein2015convex}, is to set the regularization parameter as \( \mu \propto \sqrt{\log \matrixdimension / K} \), which has been observed to be effective. In the numerical results of this paper, we choose \( \mu = 3\sqrt{\log \matrixdimension / K} \).

We proposed the joint estimation of the noise $\hat\sigma^2$ and the banded clutter--plus--noise covariance matrix $\hat\Cov$ using techniques from \cite{Jain2023} and \cite{bein2015convex}, respectively. We further laid down the condition that the estimated covariance matrix $\hat{\Cov}$ is positive definite in Theorem~\ref{Thm:SigmaPositive} and the recovery of bandsize $L$ in Theorem~\ref{Thm:BandRecover}. In the next section, we will empirically demonstrate that the normalized \gls*{scnr} is higher for the ``Banded+Spiked" covariance estimation algorithm compared to the state--of--the--art methods for an RFView simulated environment when the dimension of the matrix is greater than the number of samples ($p>K$).

\begin{algorithm}[t!]
\caption{Estimation of noise power $\hat\sigma^2$ and BCD on dual of Problem \eqref{eq:primal}.}
\emph{Inputs:} $\S,~\mu$, and weight matrices,
$\mathbf{W}^{(\ell)}$. Initialize $\mathbf{\hat{A}}^{(\ell)}=0$ for all $\ell$.\\
$\hat\sigma^2=\frac{\lambda_{\text{med}}}{\zeta_{med}}$, where $\lambda_{\text{med}}$ is the median singular value of $\S$ and $\zeta_{\text{med}}$ is the median of the \gls*{mp} distribution for the ratio $\matrixdimension/K$.\\
For $\ell=1,\ldots, \matrixdimension-1$:
  \begin{itemize}
  \item Compute $\hat{\mathbf{R}}^{(\ell)}\leftarrow
    \S-\mu\sum_{\ell'=1}^{\matrixdimension-1} \mathbf{W}^{(\ell')}\odot\mathbf{\hat A}^{(\ell')}$
  \item For $m\le\ell$, set $\mathbf{\hat A}_{s_m}^{(\ell)}\leftarrow \frac{w_{\ell m}}{\mu(w_{\ell m}^2+\max\{\hat\nu_\ell,0\})}\mathbf{\hat
      R}^{(\ell)}_{s_m}$ where $\hat\nu_\ell$ satisfies
      $h_\ell(\hat\nu_\ell)=0$, as in \eqref{eq:nu}.
  \end{itemize}
The sequence $\{\mathbf{\hat A}^{(\ell)}\}$ is a solution to \eqref{eq:dual}.
The estimated clutter--plus--noise covariance matrix is $\hat{\boldsymbol{\Sigma}}=\SampCov+\hat\sigma^2\I-\mu\sum_{\ell=1}^{\matrixdimension-1} \mathbf{W}^{(\ell)}\odot\mathbf{\hat A}^{(\ell)}$\label{alg:hat}
\end{algorithm}

\section{Numerical Results}
\label{Sec:Num}
To validate our results empirically, we use a representative STAP scenario simulator called RFView. RFView is a high-fidelity, physics based, site-specific, modeling and simulation software $\textnormal{RFView}^{\circledR}$\,\cite{RFView}. RFView uses stochastic transfer function model to simulate a scenario, where the Green's functions impulse response of the clutter and targets is computed and a real time instantiation of the RF environment is simulated. This has been extensively vetted using measured data from VHF to X band with one case documented in \cite{gogineni2022high} demonstrating the match with measured data.

The simulation scenario involves an airborne radar surveying the region over San Diego, CA, USA. The parameters for this scenario are provided in Table\,\ref{Table:Parameters} with \gls*{icm} enabled. The data-cube consists of $N_c=16$ channels, $N_s=64$ pulses and $N_r=5670$ range bins. We fix the range bin at $k_r=N_r/2$ and set the noise power to $\sigma^2=-136.48$ dB. We concatenate $q=4$ channels so the matrix dimension is $p=256$. Sample covariance matrices are generated for $K=\{128,\,161,\,203,\,256,\,322,\,406,\,512\}$ samples, with each configuration evaluated over 100 Monte Carlo simulations. To approximate the true clutter--plus--noise covariance matrix, $\Cov$, we compute a \gls*{scm} using $K=20\matrixdimension$ samples. This large sample size provides a reliable estimate of $\Cov$ under the stationary conditions of the scenario. We use the normalized \gls*{scnr} as a performance metric, which is defined as follows:
\begin{equation*}
    \text{Normalized \gls*{scnr}}=\frac{(\mathbf{y}^{H}{\hat{\Cov}}^{-1}\mathbf{y})^2}{(\mathbf{y}^{H}\Cov^{-1}\mathbf{y})(\mathbf{y}^{H}{\hat{\Cov}}^{-1}\Cov{\hat{\Cov}}^{-1}\mathbf{y})},
\end{equation*}
where $\Cov$ is the true covariance matrix, $\hat{\Cov}$ is the estimated covariance matrix and $\mathbf{y}$ is the Doppler-Azimuth steering vector. The normalized Doppler is evaluated from $[-0.5,0.5]$ with an interval of $0.05$ units and the azimuth is evaluated from $[-180^\circ,180^\circ]$ with an interval of $18^\circ$.
We estimate $\hat\Cov$ using our proposed ``Banded+Spiked" algorithm defined in the previous section with $\mu=3\sqrt{\log\matrixdimension/K}$, the \gls*{tobasco} algorithm \cite{ollila2021tobasco} with a bandsize of $L=30$ \footnote{Bandsizes were swept from 5 to 50, with 30 yielding the highest \gls*{scnr}, as shown in the results.}, the Spiked covariance estimation with the Stein Shrinkage algorithm \cite{Jain2023}, and the Diagonal Loading method\cite{abramovich1981analysis}. In Fig.\,\ref{fig:RFViewCombined}-(a), we present the average \gls*{scnr} over all Doppler and Azimuth values for each $K$. Figs.\,\ref{fig:RFViewCombined}-(b) and (c) show results for $K=161$, displaying the average \gls*{scnr} over Azimuth and Doppler, respectively. 

%\subsection{Empirically showing Unbiased}

Using the scenario simulated in RFView, we demonstrated that the proposed ``Banded+Spiked" clutter-plus-noise covariance estimation algorithm achieves high \gls*{scnr}. The proposed algorithm outperforms TABASCO, which captures only the tapered structure of the clutter-plus-noise covariance matrix without accounting for the noise floor. The spiked covariance model alone lacks the ability to represent the banded structure effectively. Notice that in the diagonal loading method, a heuristic approach for stabilizing the \gls*{scm} in low-data settings, yields sub-optimal performance compared to the proposed algorithm.
\begin{table}[h!]
\centering
\begin{tabular}{|p{3cm}|p{3cm}|} % Adjust column widths as needed
\hline
\textbf{Parameter} & \textbf{Specification} \\
\hline
\textbf{Radar Configuration} & Monostatic \\
\hline
\textbf{Carrier Frequency} & 10 GHz \\
\hline
\textbf{Bandwidth} & 5 Mhz \\
\hline
\textbf{Pulse Repetition Frequency (PRF)} & 1.1 kHz \\
\hline
\textbf{Number of Pulses} & 64 \\
\hline
\textbf{Platform} & Airborne, 1000 m altitude, 100 m/s speed, heading North \\
\hline
\textbf{Antenna Array} & 48 horizontal elements, 5 vertical elements, $\lambda/2$ spacing \\
\hline
\textbf{Transmitter (Tx) Array} & Single channel \\
\hline
\textbf{Receiver (Rx) Array} & 16 channels, pre-steered to aimpoint \\
\hline
\textbf{Wind Speed} & 14\,km/h\\
\hline
\end{tabular}
\caption{Airborne Radar Parameters}
\label{Table:Parameters}
\end{table}
\begin{figure}[htbp]
    \centering
    % Subfigure 1: Normalized SINR vs. Number of Samples
    \begin{subfigure}{0.5\textwidth}
        \includegraphics[width=\linewidth]{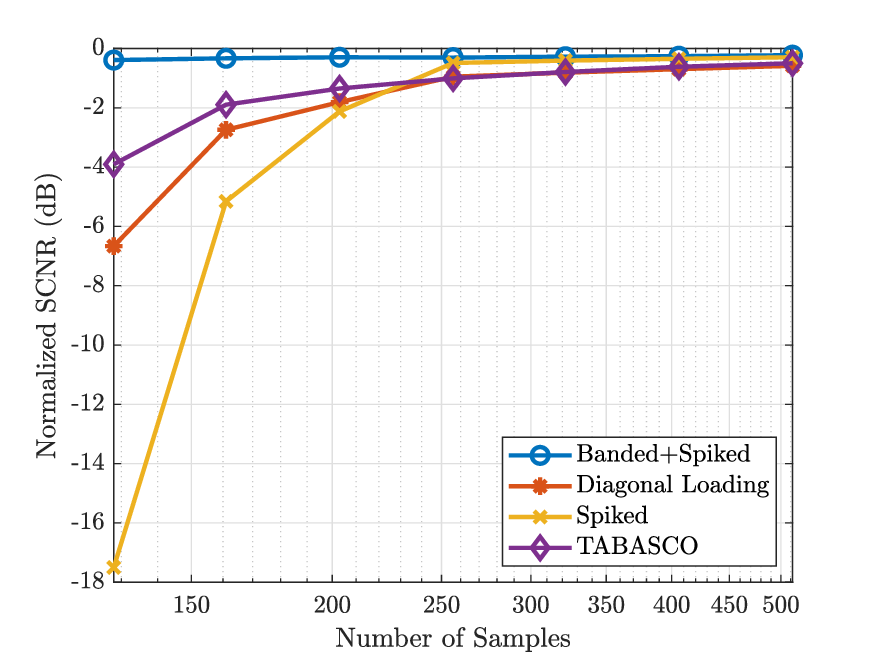}
        \caption{Matrix Dimension $\matrixdimension=256$}
        \label{fig:RFViewSamples}
    \end{subfigure}
    \hfill
    % Subfigure 2: Normalized SINR vs. Doppler
    \begin{subfigure}{0.5\textwidth}
        \includegraphics[width=\linewidth]{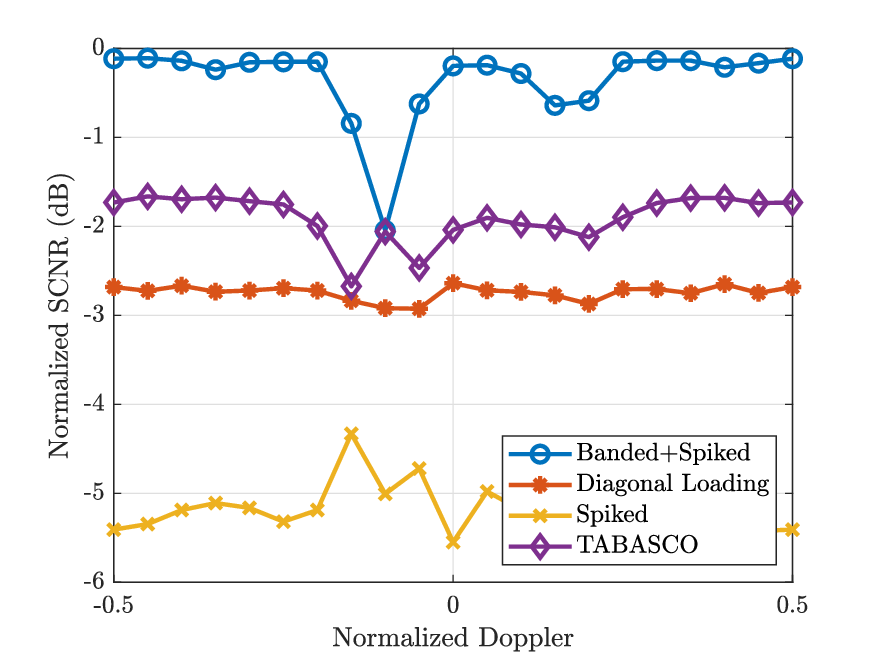}
        \caption{Matrix Dimension $\matrixdimension=256$, Samples $K=161$}
        \label{fig:RFViewDoppler}
    \end{subfigure}
    \hfill
    % Subfigure 3: Normalized SINR vs. Azimuth
    \begin{subfigure}{0.5\textwidth}
        \includegraphics[width=\linewidth]{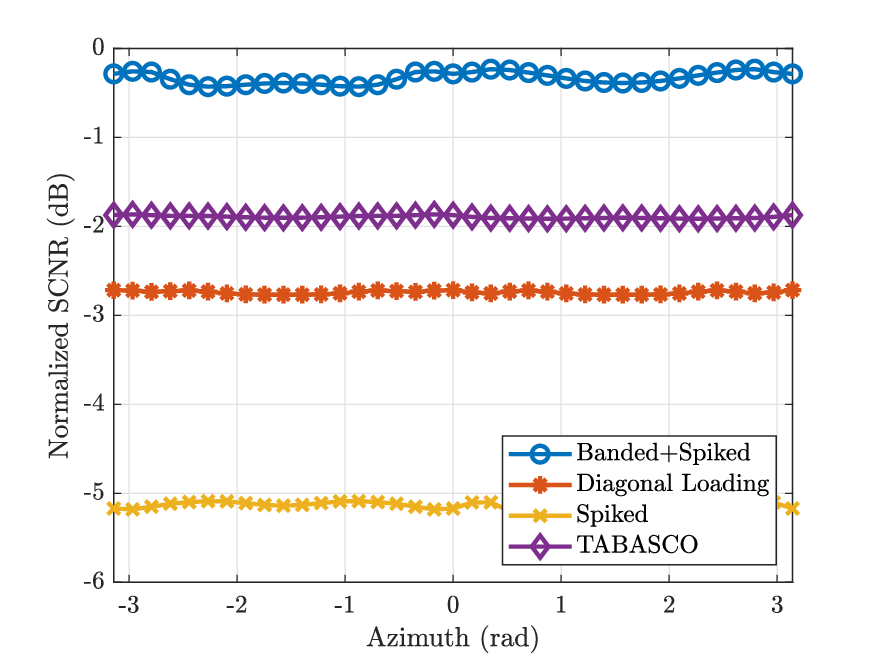}
        \caption{Matrix Dimension $\matrixdimension=256$, Samples $K=161$}
        \label{fig:RFViewAzimuth}
    \end{subfigure}

    \caption{The proposed ``Banded+Spiked" estimation algorithm consistently achieves higher \gls*{scnr} in the low-data regime, where the number of samples is less than the matrix dimension ($K < \matrixdimension$).}
    \label{fig:RFViewCombined}
\end{figure}
\section{Conclusion}
We modeled the clutter--plus--noise covariance matrix using a ``Banded+Spiked" covariance structure. We proposed a convex upper bound, a Frobenius norm minimization with nonsmooth convex constraints, for the \gls*{mle} problem of the ``Banded+Spiked" covariance matrix estimation, for both cases: when $p>K$ and $p\leq K$. We proposed a variational inequality based bound to quantify the closeness of the convex relaxation. Using techniques from our previous work \cite{Jain2023} and building upon algorithm from \cite{bein2015convex}, we developed an algorithm to estimate the ``Banded+Spiked" covariance matrix effectively. We laid the conditions for the positive definiteness of the estimated clutter--plus--noise covariance matrix and the bandsize recovery. Numerical results using RFView demonstrated that the proposed algorithm achieves higher \gls*{scnr} compared to several state-of-the-art methods. 

In future work, we will explore the target detection performance of the proposed estimator. We plan to investigate a generalized block-banded structure for the clutter--plus--noise covariance matrix in a multi-channel setting. We believe that the block-banded structure can be effectively handled by appropriately modifying the weighting matrix \( \mathbf{W} \) defined in Sec.\,\ref{Sec:Algo}-\ref{Sec:NoiseEstPlusBanded}\ref{Sec:BandedEst}. Additionally, we plan to develop a tighter convex relaxation that is globally Lipschitz with respect to the variable. In this context, variational-based bounds will serve as a powerful analytical tool.\\

\noindent\textbf{Acknowledgement:} We express our gratitude to Professor Marten Wegkamp, in the Department of Statistics and Data Science at the Cornell University, for useful discussion regarding the banded covariance estimation which guided the development of the paper.
\appendix
\section*{Proof for Upper Bound (Theorem\,\ref{Thm:UpperBound})}
\noindent \textbf{Case 1: $\matrixdimension\leq K$}
\begin{align*}
    &\tr(\Cov^{-1}\SampCov)-\log\det(\Cov^{-1}\SampCov)-\matrixdimension=\tr(\I+\Cov^{-1}\Delta)\\
    &-\log\det(\I+\Cov^{-1}\Delta)\approx \matrixdimension+\tr(\Cov^{-1}\Delta)\\&-(\tr(\Cov^{-1}\Delta)
    +\frac{1}{2}\tr((\Cov^{-1}\Delta)^2))-\matrixdimension\\
    &=\frac{1}{2}\left\|\Cov^{-1}(\SampCov-\Cov)\right\|_{F}^{2}.
\end{align*}
Using the Frobenius norm inequality,
\begin{align*}
    \frac{1}{2}\left\|\Cov^{-1}(\SampCov-\Cov)\right\|_{F}^{2}\leq \frac{1}{2}\left\|\Cov^{-1}\right\|_2^2\left\|\Cov-\SampCov\right\|_{F}^{2}
\end{align*}
Since, we know that the 2-norm $\left\|\Cov^{-1}\right\|_2^2\approx\sigma^{-4}$,
\begin{equation}
    F(\SampCov; \Cov)\leq \left\|(\Cov-\SampCov)/\sigma^2\right\|_{F}^{2}+\mathrm{const.}.
\end{equation}
\textbf{Case 2: $K<\matrixdimension$}\\
\begin{comment}
In this case, we first define the complex singular Wishart distribution, since $|\SampCov|=0$ so the complex Wishart distribution is not defined. The expression for Complex Singular Wishart Distribution is:
\[f(\SampCov;\Cov)=\frac{\pi^{K(K-\matrixdimension)|\boldsymbol{\Lambda}|^{K-\matrixdimension}}}{\mathcal{C}\Gamma_{K}(K)|\Cov|^{K}}\exp(-K\tr(\Cov^{-1}\SampCov))\]
where $\mathcal{C}\Gamma_{K}(K)$ is the complex multivariate Gamma function and $\boldsymbol{\Lambda}=\diag(\lambda_1,\,\lambda_2,\,\dots,\,\lambda_{K})$, the first $K$ positive singular values of $\SampCov$. In this case, as we maximize the log-likelihood function or minimize the negative Log-Likelihood function.
\end{comment}
Observe that if we work with the following expression:
\begin{equation}
\label{eq:loglikelihoodCSWD}
    \mathcal{L}(\Cov;\SampCov)=\log\det(\Cov)+\tr(\Cov^{-1}\SampCov),
\end{equation}
we can come up with some Frobenius norm-based upper bound. Using the Cauchy-Schwarz inequality:
\begin{align*}
   &\tr(\Cov^{-1}\SampCov)=\tr(\Cov^{-1}(\SampCov-\Cov)+\I)\\
   &=\tr(\Cov^{-1}(\SampCov-\Cov))+\matrixdimension\\
   &\leq\left\|\Cov^{-1}\right\|_{F}\left\|\Cov-\SampCov\right\|_{F}+\matrixdimension\\
   &\leq \sqrt{\matrixdimension}\left\|\Cov^{-1}\right\|_{2}\left\|\Cov-\SampCov\right\|_{F}+\matrixdimension\\
   &\sim\sqrt{\matrixdimension}\left\|\Cov-\SampCov\right\|\sigma^{-2}+\matrixdimension.
\end{align*}
The term $\log\det(\Cov)\leq \matrixdimension\log\lambda_{\max}(\Cov)\leq \matrixdimension\log c$. We bound this term under the assumption that we are in a finite energy, finite power system with finite thermal noise making it independent of the term $\Cov$. Therefore, for case $K<\matrixdimension$, the minimization of Frobenius norm $\|\Cov-\SampCov\|_{F}^2$, is the upper bound for the minimizing the negative log-likelihood function defined in \eqref{eq:loglikelihoodCSWD}.
\section*{Proof of Theorem\,\ref{Thm:EpsilonDerive}}
\begin{enumerate}
    \item For the case $p\leq K$, we have the convex relaxation $\upperbound(\SampCov;\;\Cov)$ in \eqref{eq:UpperBoundHigh} and the non-convex function $F(\SampCov\;;\Cov)$ defined in \eqref{eq:LogLikelihood}. It is trivial that for a given $m=1$, given that $\underset{\Cov}{\mathrm{argmin}}\left\|\Cov-\SampCov\right\|_{F}^{2}=\SampCov$,
    \[\left\|\Cov-\SampCov\right\|_{F}^{2}\geq 0+m\left\|\Cov-\SampCov\right\|_{F}^{2}, \]
    therefore $\gamma=2$.
    Now, from Theorem\,\ref{Thm:UpperBound}, we know that $\upperbound(\SampCov;\;\Cov)\geq F(\SampCov;\;\Cov)$, therefore $\sup_{\Cov}|\upperbound(\SampCov;\;\Cov)-F(\SampCov;\;\Cov)|$ can be derived by computing the upper bound for $\upperbound(\SampCov;\;\Cov)$ and lower bound for $F(\SampCov;\;\Cov)$. 
    
    By triangle inequality, $\left\|\Cov-\SampCov\right\|_{F}^{2}\leq(\left\|\Cov\right\|_{F}+\left\|\SampCov\right\|_{F})^2\leq \matrixdimension(c+\lambda_{\max}(\SampCov))^2$. Therefore,
    \begin{align*}
        \upperbound(\Cov)&=\left\|(\Cov-\SampCov)/\sigma^2\right\|_{F}^2+\const\\
        &\leq \frac{\matrixdimension}{\sigma^4}(c+\lambda_{\max}(\SampCov))^2+\const.
    \end{align*}
    Given that $\lambda_{\max}(\Cov)\leq c$, therefore $\tr(\Cov^{-1}\SampCov)\geq\frac{\tr(\SampCov)}{c}$. Similarly, we know that $\lambda_{\min}(\Cov)=\sigma^2$, so $\log\det\,\Cov\SampCov^{-1}\geq -\matrixdimension\log\frac{\lambda_{\max}(\SampCov)}{\sigma^{2}}$. Therefore,
    \begin{align*}
        F(\Cov;\SampCov)&=\tr(\Cov^{-1}\SampCov)+\log\det(\Cov\SampCov^{-1})-\matrixdimension+\mathrm{const.,}\\
        &\geq \frac{\tr(\SampCov)}{c}-\matrixdimension\log\frac{\lambda_{\max}(\SampCov)}{\sigma^{2}}-p+\const.
    \end{align*}
    The constant terms will cancel out, so
    \begin{align*}
        &\sup_{\Cov}|\upperbound(\Cov;\;\SampCov)-F(\Cov;\;\SampCov)|\leq \frac{\matrixdimension}{\sigma^4}(c+\lambda_{\max}(\SampCov))^2\\ &+\matrixdimension\log\frac{\lambda_{\max}(\SampCov)}{\sigma^{2}}+p-\frac{\tr(\SampCov)}{c}.
    \end{align*}
    Therefore,
    \[\epsilon=\frac{\matrixdimension}{\sigma^4}(c+\lambda_{\max}(\SampCov))^2+\matrixdimension\log\frac{\lambda_{\max}(\SampCov)}{\sigma^{2}}+p-\frac{\tr(\SampCov)}{c}.\]
    \item For the case $\matrixdimension>K$, it is trivially established that for $m=1$, $\gamma=1$. Using the similar triangle inequality from previous case, we have 
    \begin{align*}
        \upperbound_{S}(\Cov;\;\SampCov)&=\sqrt{\matrixdimension}\left\|(\Cov-\SampCov)/\sigma^2\right\|_{F}+\matrixdimension\nonumber\\&+\matrixdimension\log\,c-\log\det\boldsymbol{\Lambda}+\const\,\\
        &\leq \frac{p}{\sigma^2}(c+\lambda_{\max}(\SampCov))+\matrixdimension\log\,c\\ &-\log\det\boldsymbol{\Lambda}+\const\,.
    \end{align*}
    \begin{align*}
F_{\mathrm{S}}(\Cov;\SampCov)&=\tr(\Cov^{-1}\SampCov)+\log\det(\Cov)\\
&-\log\det(\boldsymbol{\Lambda})+\mathrm{const.}\\
&\geq \frac{\tr(\SampCov)}{c}+\matrixdimension\log\,\sigma^2\\ &-\log\det(\boldsymbol{\Lambda})+\mathrm{const.}
\end{align*}
the constants will cancel out, so,
\begin{align*}
        &\sup_{\Cov}|\upperbound(\Cov;\;\SampCov)-F(\Cov;\;\SampCov)|\leq \frac{\matrixdimension}{\sigma^2}(c+\lambda_{\max}(\SampCov))\\ &+\matrixdimension\log\frac{c}{\sigma^{2}}+\matrixdimension-\frac{\tr(\SampCov)}{c}.
    \end{align*}
    Therefore,
    \[\epsilon=\frac{\matrixdimension}{\sigma^2}(c+\lambda_{\max}(\SampCov))+\matrixdimension\log\frac{c}{\sigma^{2}}+\matrixdimension-\frac{\tr(\SampCov)}{c}.\]
\end{enumerate}
\section*{Proof of Theorem\,\ref{Thm:SigmaPositive}}
\noindent Let $\mathbf{u}$ be the eigenvector such that $\mathbf{u}^{H}\hatP\mathbf{u}=\lambda_{\min}(\hatP)$, where $\hatP=\hat\Cov-\hat\sigma^2\I$. From the proof of Theorem~\ref{Thm:psd} (\cite[Theorem\,10]{bein2015convex}),
\[\lambda_{\min}(\hatP)=\mathbf{u}^{H}\Cov\mathbf{u}-\mathbf{u}^{H}(\Cov-\hatP)\mathbf{u}\]
\[\geq \sigma^2-\|\Cov-\hatP\|_{\mathrm{op}},\]
where $\|\cdot\|_{\mathrm{op}}$ is the operator norm. By \cite[Theorem~9]{bein2015convex}, $\|\Cov-\hatP\|_{\mathrm{op}}\leq C'K\sqrt{\log\matrixdimension/K}$, on the set $\mathcal{A}_x$ w.p. $1-c_1/\matrixdimension$. Therefore, if $\sigma^2\geq C'K\sqrt{\log\matrixdimension/K}$ for some positive constant $C'>0$, then $\lambda_{\min}(\hat\Cov)\geq \hat\sigma^2$ w.p. $1-c_1/\matrixdimension$ on the set $\mathcal{A}_x$.

\bibliographystyle{ieeetr}
\bibliography{references.bib}
\end{document}